\documentclass[10pt,journal,cspaper,compsoc]{IEEEtran}

\usepackage{comment}
\usepackage{amsbsy}
\usepackage{amssymb}
\usepackage{mathrsfs}
\usepackage{latexsym}
\usepackage{enumerate}
\usepackage{amsmath}
\usepackage{amsfonts}
\usepackage{float}
\usepackage{wrapfig}
\usepackage{lipsum}
\usepackage{caption}
\usepackage{color}
\usepackage{authblk}

\newcommand{\dk}[1]{{#1}}

\usepackage[noend]{algpseudocode}
\usepackage[ruled]{algorithm}

\newcommand{\remove}[1]{}
\newcommand{\qedsymb}{\hfill{\rule{2mm}{2mm}}}

\floatname{algorithm}{Protocol}
\floatstyle{plain}
\newfloat{myalgo}{tbhp}{mya}

{\begin{myalgo}[#1]
\centering
\begin{minipage}{6.0in}
\begin{algorithm}[H]}%
{\end{algorithm}
\end{minipage}
\end{myalgo}}

\newcommand{\BTD}{BTD}
\newcommand{\SaSmod}{BTD\_Construct}

\newcommand{\eps}{\varepsilon}
\newcommand{\m}{\mathcal}

\newcommand{\cT}{{\mathcal T}}
\newcommand{\cN}{{\mathcal N}}

\newcommand{\INT}{{\mathbb Z}}
\newcommand{\NAT}{{\mathbb N}}
\newcommand{\dist}{\text{dist}}

\newtheorem{theorem}{Theorem}
\newtheorem{lemma}{Lemma}
\newtheorem{proposition}{Proposition}
\newtheorem{corollary}{Corollary}

\begin{document}

\title{Multi-Broadcasting under the SINR Model}

\author[1]{Sai~Praneeth~Reddy}
\author[2]{Dariusz R.~Kowalski}
\author[3]{Shailesh~Vaya}
\affil[1]{Department of Computer Science\\
Indian Institute of Technology Delhi, India\\
\textit{saipraneet@gmail.com}}
\affil[2]{Department of Computer Science
\\University of Liverpool, UK
\textit{darek@liverpool.ac.uk}}
\affil[3]{Xerox Research Centre India\\
Bangalore, India\\
\textit{shailesh.vaya@xerox.com}}






\IEEEcompsoctitleabstractindextext{%
\begin{abstract}
  We study the multi-broadcast problem in multi-hop wireless networks under the SINR model 
deployed in the 2D Euclidean plane. In multi-broadcast, there are $k$ initial rumours, potentially belonging to different nodes, that must be forwarded to all $n$ nodes of the network. 
\remove{
Furthermore, in each round a node can only transmit a small message that could contain at most one initial rumor and $O(\log n)$ control bits. In order to be successfully delivered to a node, transmissions must satisfy the (Signal-to-Inference-and-Noise-Ratio) SINR condition and have sufficiently strong signal at the receiver. 
}
We present deterministic algorithms for multi-broadcast for different settings that reflect the different types of knowledge about the topology of the network
available to the nodes: (i) the whole network topology (ii) their own coordinates and coordinates of their neighbors (iii) only their own coordinates, and (iv) only their own ids and the ids of their neighbors. For the former two settings, we present solutions that are scalable with respect to the diameter of the network and 
the polylogarithm of the network size, i.e., $\log^c n$ for some constant $c> 0$, while the solutions for the latter two have round complexity that is superlinear in the number of nodes. The last result is of special significance, as it is the first result for the SINR model that does not require nodes to know their coordinates in the plane (a very specialized type of knowledge), but intricately exploits the understanding that nodes are implanted in the 2D Euclidean plane.
\end{abstract}




\begin{keywords}
  Wireless communication, SINR model, 
  Distributed algorithms, Centralized algorithms, Deterministic Algorithms
\end{keywords}
}

\maketitle



\section{Introduction}
\label{intro}

\remove{
\textbf{Things to do}
\begin{enumerate}
\item{} Cite works which do not use knowledge of coordinates: wagner, chinese, kuhn and discuss
\item{} Discuss weak devices and their physical motivations and their meaningfullness. elaborately cite the weak devices assumption in other works.
\end{enumerate}
}

  We consider the Signal-to-Inference-and-Noise-Ratio (SINR) model for communication in ad-hoc wireless networks. The wireless network consists of $n$ {\em stations}, also called {\em nodes}. 
Each node is assigned a unique ID in the range $\{1,\dots,N\}$, which is also called its label. Furthermore, all nodes are assumed to lie in a 2-dimensional space with Euclidean metric and have uniform transmission powers. The communication graph of the network is a graph defined on the nodes; an edge $(u,w)$ exists in communication graph if node $w$ can successfully receive the message transmitted by $u$ when no other node is simultaneously transmitting. The diameter of the communication graph is referred to by $D$ and the maximum degree by $\Delta$. 

  In the {\em multi-broadcast} problem, there is a set $K$ of {\em source nodes} that are active in the beginning 
of the protocol, each of them with unique packet (also called rumor).
The goal is to deliver all rumors stored in source nodes to all nodes in the network.  
In this work we study the multi-broadcast problem in SINR model with unit-size restriction on message size.
The {\em unit-size} restriction on the size of each message sent says that each message must contain at most one original rumor and $O(\lg n)$ additional control bits. 
Furthermore, we assume the {\em non-spontaneous} setting, in which all nodes except the nodes in set $K$ are asleep at the beginning. The asleep nodes cannot transmit a message till they receive a message from some neighboring node. We consider {\em round complexity} as the sole measure for comparing efficiency of the distributed protocols developed in this work.

\subsection{Our results}
\label{ourresults}


In this work, we present a fairly comprehensive and rigorous study of the multi-broadcast problem in the context of SINR model 
in several settings. 
\begin{enumerate}
\setlength{\leftskip}{-0.3cm}
\item{} For the centralized setting, 
\dk{where nodes have full knowledge about the topology of the network,}
we present a deterministic algorithm that runs in $O(D + k \lg\Delta)$ rounds. We also present an algorithm sensitive to the granularity of the network $g$, where granularity is defined as the maximum transmission range times the inverse of the minimum distance between any two stations. It accomplishes 
multi-broadcast in $O(D + k + \lg g)$ rounds.
\item{} For the setting in which nodes know only their own coordinates and the coordinates of their neighbors, we present a deterministic algorithm with round complexity $O(D \lg^2 n + k \lg \Delta)$.
\item{} For the setting in which the nodes are given only their own coordinates we present a deterministic algorithm with a round complexity of $O((n+k) \lg n)$.

\item{} Our most interesting result is for the setting when nodes know only their own labels and the labels of their neighbors, besides the standard knowledge of parameters $n,N,$ $k,D,\Delta$.
No deterministic results have been known for this setting in the literature. 
It seems somewhat hard to fathom that one could develop a fast algorithm for this setting considering that nothing is known about the underlying geometric positions of the node.

  Single source depth first search on the network is easy to conduct on the network using the neighborhood information. However, the problem becomes challenging when multiple sources concurrently start such a search because of their transmissions can interfere with each other.

  We present a deterministic protocol for this setting, which exploits the 
fact
  that the nodes are embedded in the 2-Dimensional Euclidean plane without explicitly utilizing the actual  coordinates of the nodes. Our algorithm 
runs in time $O((n+k) \lg n)$, and it
stitches together a few ideas developed for \textit{Breadth-Then-Depth} search trees and their efficient distributed construction under the SINR model, as well as efficient token elimination.
\end{enumerate}

\section{Model}
\label{prelim}
  We consider a wireless network which consists of $n$ {\em nodes} deployed in a two-dimensional Euclidean plane. The {\em Euclidean metric} on the plane in which the nodes are embedded is denoted $\dist(\cdot,\cdot)$. The {\em transmission power} of station $v$ is noted by $P_v$ and is a positive real number.

  The Signal-To-Noise-Interference-Ratio Model, aka SINR model, is characterized by three parameters: path loss $\alpha> 2$, ambient noise $\cN > 0$, and threshold $\beta\ge 1$.%
\footnote{For simplicity, in the analysis we assume $\beta=1$; this can be easily scaled up to any $\beta\ge 1$.}

  For a set of stations, transmitting in the same round, $\cT$ the success of transmission from a transmitting node $v$ to a receiving node $u$ depends on the transmission power of $u$ reaching node $v$ and the inference due other concurrent transmissions in the neighborhood. In particular, this signal strength to noise ratio referred to as $SINR(v,u,\cT)$ is defined as follows:

\begin{equation}\label{e:sinr}
SINR(v,u,\cT)
=
\frac{P_v\dist(v,u)^{-\alpha}}{\cN+\sum_{w\in\cT\setminus\{v\}}P_w\dist(w,u)^{-\alpha}}
\end{equation}

A station $u$ successfully receives a message from a station $v$ in a round if it is true that 
$v \in \cT$, $u \notin \cT$ and the following conditions hold true:

\begin{itemize}
\item[a)]
$P_v\dist^{-\alpha}(v,u)
\geq
(1+\eps)\beta\cN$

\item[b)]
$SINR(v,u,\cT)
\ge
\beta$ ,
\end{itemize}
where $\cT$ denotes the set of stations which are transmitting concurrently and $\eps>0$ is a fixed {\em signal sensitivity parameter} of the model.%


\remove{gs
In the paper, we assume for the sake of
clarity of presentation that $\beta=1$ and
$\cN=1$.
These assumptions can be dropped without harming the asymptotic performances of
the presented algorithms and lower bounds formulas.
}

\paragraph{Ranges and uniformity}
The {\em communication range} $r_v$ of a station $v$ is defined as the radius of the
ball in which a message transmitted
by the station is received, given that no other station transmits at the same time.
In this paper, only uniform networks are considered with $r_v=r$.
\remove{gs
For clarity of presentation
we
assume
that all powers are equal to $1$, i.e., $P_v=1$ for each $v$.
This assumption
can be dropped without changing asymptotic formulas.
Under these assumptions, $r_v=r$
for each station $v$.
}
The {\em range area} of a station $v$ is defined to be the ball of radius $r$ which is centered at $v$.

\paragraph{Communication graph and graph notation}
  The {\em communication graph} $G(V,E)$ of a given network consists of all network nodes and edges $(v,u)$ such that $u$ is in the range of $v$. The communication graph is also called the {\em reachability graph}. For uniform networks, the communication graph is symmetric.
  The {\em neighborhood} of a node $u$ is defined as the set of all 
neighbors of $u$ in $G$, i.e. the set $\{w\,|\, (w,u)\in E(G)\}$.

  The {\em graph distance} from a node $v$ to node $w$ is equal to the length of a shortest path from $v$ to $w$ in the communication graph, where the length of a path is equal to the number of edges contained in it.

  $\Delta$ is used to denote the maximum degree of a node in the communication graph.

\remove{gs-potem
Assume now that many stations transmit in parallel.
We say that a station $v$ transmits {\em $c$-successfully} in a round
$t$ if $v$ transmits a message in round $t$ and this message is heard by
each station $u$ in the Euclidean distance
at most
$c$ from $v$. We say that a station
$v$ transmits {\em successfully} in round $t$ if each of its neighbors in the communication graph can hear its message.
Finally, $v$ transmits {\em successfully} to $u$ in round $t$ if $v$ transmits
a message and $u$ receives this message in round $t$.
}

\paragraph{Synchronization}
  It is assumed that the protocols work synchronously in rounds. Each station can act
either as a sender or as a receiver in a given round. We do not assume ticking of global clock. The reader may note that it is easy to guarantee the same clock in all nodes by propagating its current reading (or round counter) piggybacked to the transmitted messages (it adds only $O(\lg n)$
additional bit to the message).

\paragraph{Carrier sensing}
  We consider the model {\em without carrier sensing}. That is,
a station $u$ has no other feedback from the wireless channel 
than
receiving or not successfully receiving a message in a round $t$.

\paragraph{Knowledge of stations}
Each station is assigned a unique ID from the set $[N]$,%
\footnote{We denote $[i]$ to refer to the set $\{1,2,\ldots,i\}$ and $[i,j]$ to the set $\{i,i+1,\ldots,j\}$ for $i,j\in\mathbb N$.}
where $N$ is a polynomial in $n$.

  Stations may know their locations and locations of other nodes or their neighborhood besides parameters $n$, $N$. This exact specification of the setting is clearly stated for every setting we study.
  Some subroutines use the parameter granularity $g$, which is defined as $r$ times the
inverse of the minimum distance between any two stations (c.f.,~\cite{EmekGKPPS09}).


\paragraph{Multi-broadcast problem and complexity parameters}
  In the broadcast problem, 
there is one distinguished node, called the {\em source}, which initially holds a piece of information (also called a source message or a broadcast message or a rumor). The goal is to deliver this message to all other nodes in the network.

In the multi-broadcast problem, a set $K$ of stations hold $k$ rumors in total,
which are to be disseminated to the rest of the network. 
We do not assume that $|K|=k$, therefore $k$ could be seen as an upper bound on $|K|$.

The {\em round complexity} denotes the number of communication rounds for which a protocol is executed  before accomplishing the task (multi-broadcast) in the worst case.

\paragraph{Messages and initialization of stations other than source}

We assume that a single message sent in the execution of any algorithm
can carry a single rumor and a number of control bits, which is upper bounded by some 
$O(\lg n)$. This model is called in the literature a {\em unit-size message} model.
We consider the {\em non-spontaneous wake-up} setting, in which 
only the initial subset $K$ of rumor sources are awake and other nodes have to receive a message 
in order to start their participation in the protocol (prior to this they are idle and only listening to the wireless medium).

\subsection{Previous and Related Results}

  In \cite{JK12}, the authors consider the model of (uniform power) weak devices and designed distributed deterministic algorithms for building a backbone structure in $O(\Delta\mbox{ polylog } n)$ rounds.
	Unlike in our setting, in~\cite{JK12} it was assumed that all nodes simultaneously start building the backbone (so called {\em spontaneous wake-up} setting). In another recent work, \cite{JKS-ICALP-13}, a non-spontaneous wake-up was assumed, as in our paper. It studied deterministic {\em single} broadcast, developed several algorithms (amongst others $O(n\log n)$ algorithm with knowledge of only own coordinates and $O(D\log^2 n)$ if nodes know also coordinates of their neighbors), and proved lower bounds separating models with and without local knowledge as well as implying that there is an extra cost payed due to lack of synchronization (when comparing to~\cite{JK12}).
  Both these papers assumed knowledge of coordinates. Our study, in turn, goes further in two aspects: first, we study more general problem of multi-broadcast, and second, we analyze the impact of knowledge of coordinates on algorithm's performance, which brings a new perspective of wireless devices not equipped with GPS.

  Deterministic broadcasting with {\em strong devices} (i.e., not restricted by the fact that the signal must be sufficiently strong in order to be noticed) can be done in $O(D \lg^2 n)$ with the knowledge of coordinates. This was established in~\cite{JKS-FCT-13}. From these results, it is inferred that there is a complexity gap between the two models (i.e., weak and strong devices) for broadcast problem. Slightly faster randomised solutions were developed in~\cite{JKRS-DISC-13}, and other in slightly different models~\cite{DGKN13, YuHWYL13} (in the latter, the setting without knowledge of coordinates was considered and the complexity raised by the polylogarithm of the granularity).

\remove{
Deterministic data aggregation and a few other problems have also been studied for SINR-based models in ad hoc setting, c.f.,~\cite{HobbsWHYL12}. 
{\em Local} broadcasting, in which the requirement is to inform only their neighbors in the corresponding reachability graph, was studied in~\cite{YuWHL11}. The setting allowed algorithm to control power, so that stations could transmit with any power smaller than the maximal one in order to avoid collisions.
Randomized solutions for contention resolution and local broadcasting were given in~\cite{GoussevskaiaMW08} and~\cite{KV10} respectively.
Multiple Access Channel properties were also recently studied under the SINR model, c.f.,~\cite{RichaSSZ}.
}

 There is a vast amount of work on centralized algorithms under the SINR model, for which the  most studied problems include connectivity, capacity maximization, link scheduling etc.; The reader is directed to the survey~\cite{WatSurv} for recent results. 

\remove{
\item
z prac o connectivity w SINR, podaje cos co traktuje o uniform power:
\cite{AvinLPP09} (stala liczba kolorow, ale stacje tylko w wezlach gridu);
\cite{AvinLP09} (o tym, ze uniform niewiele gorsze od nonuniform);
\item
w surveyu Wattenhoffera i in. jest cala kolekcja wynikow na temat one-slot scheduling
i multi-slot scheduling offline (\textbf{scentralizowany}) dla modelu uniform: NP-zupelnosc, algorytmy
aproksymacyjne... a z algorytmow rozproszonych wymieniaja glownie:
\cite{GoussevskaiaMW08} o local broadcasting zrandomizowanym (``each node performs a successful local
broadcasting in time proportional to the number of neighbors
in its physical proximity''); \cite{LebharL09} traktuje o uniform (udg): nie doczytalem dokladnie,
ale chodzi o zrandomizowana symulacje collision-free (?) UDG w modelu SINR przy jednostajnym rozkladzie
wierzcholkow w ustalonym kwadracie...
\end{itemize}
}

\paragraph{Radio network model}
  In the model of radio networks, a transmitted message is successfully received if there are no other simultaneous transmissions from the neighbors of the receiver in the reachability graph.
The model does not take into account the real strength of the received signals and the signals from outside of the close proximity, however some techniques related to restricting {\em local} interference may be similar.
In the geometric ad hoc setting, Dessmark and Pelc~\cite{DessmarkP07} were the first who studied this problem. They analyzed the impact of local knowledge, which is defined as the range within which stations can discover the nearby stations.
%
Emek et al.~\cite{EmekGKPPS09} presented a broadcast algorithm working in time $O(Dg)$ in Unit Disc Graphs (UDG) radio networks with eccentricity $D$ and granularity $g$. In Emek et al.~\cite{EmekKP08} proved a matching lower bound $\Omega(Dg)$.
%

In the {\em graph-based} model of radio networks, stations may not be explicitly deployed in a metric
space. The fastest $O(n\log(n/D))$-round algorithm was developed by Kowalski~\cite{Kow-PODC-05} and almost a matching lower
bound was given by Kowalski and Pelc~\cite{KP-DC-05}, who also studied fast randomized
solutions (in parallel with~\cite{CzumajRytter-FOCS-03}).
The above results hold without the assumption of local knowledge.
When local knowledge is assumed, Jurdzinski and Kowalski~\cite{JK-OPODIS-12} showed
a lower bound $\Omega(\sqrt{Dn\log n})$ on the number of rounds and an algorithm
of relatively close round complexity $O(D\sqrt{n}\log^6 n)$.
Multi-broadcast with unit size messages has also been studied intensively for the ad-hoc radio networks model, c.f., \cite{CKPR-ICALP-11}.

\subsection{Technical Preliminaries}

For the considered non-spontaneous wake-up setting, observe that a round counter could be easily maintained by already informed nodes by passing it along the network with the transmitted messages.
In this sense, all algorithms can be assumed to have a global clock. 
Note also that for $K$ being the set of all nodes, the obtained setting is the spontaneous wake-up one.

In the multi-broadcast protocols, we explicitly specify the details of the message that is transmitted by a node.

A station $v$ transmits {\em successfully} (or to station $u$) in round $t$ if each of its neighbors (station $u$) in the communication graph can hear its message.

\paragraph{Grids}
The notations for grids are taken from \cite{JKS-ICALP-13}. For a given a parameter $c>0$, we define 
a partition of the $2$-dimensional space into square boxes of size $c\times c$ by the grid $G_c$, in such a way that: all boxes are aligned with the coordinate axes, point $(0,0)$ is a grid point, each box includes its left side without the top endpoint and its bottom side without the right endpoint and does not include its right and top sides.
  $(i,j)$ is the coordinate of the box with its bottom left corner located at $(c\cdot i, c\cdot j)$, for $i,j\in \INT$. A box with coordinates $(i,j) \in \INT^2$ is denoted as $C(i,j)$.

As found in \cite{DessmarkP07,EmekGKPPS09}, the {\em grid} $G_{r/\sqrt{2}}$ is very useful in the design of the algorithms for UDG (unit disk graph) radio networks, where
$r$ is equal to the range of each station.
This is because $r/\sqrt{2}$ is the largest parameter of a grid such that each
station in a box is in the range of every other station in that box.

Fix $\gamma=r/\sqrt{2}$, where $r=(1+\eps)^{-1/\alpha}$ is the transmission range, and call $G_{\gamma}$ the {\em pivotal grid}. If not stated otherwise, we shall be referring to (boxes~of)~$G_{\gamma}$.

\paragraph{Schedules}
  A (general) {\em broadcast schedule} $\mathcal{S}$ of length $T$ wrt $N\in\NAT$ is a mapping from the set of plausible labels $[N]$ to binary sequences of length $T$. A station with identifier $v \in [N]$ {\em follows} the schedule $\m{S}$ of length $T$ if $v$ transmits a message in round $t$ of that period iff the position $t\mod T$ of $\m{S}(v)$ is equal to $1$.

  A geometric broadcast schedule $S$ of length $T$ with parameters $N, \delta \in N$, $(N, \delta)$-gbs for short, is a mapping from $[N] * [0, \delta-1]^2$ to binary sequences of length $T$. $v$ follows $(N, \delta)$-gbs $S$ for the grid $G_c$  in a fixed period of time, when $v$ transmits a message in round $t$ of that period iff $t^{th}$ position of $S(v, i \mod \delta, j \mod \delta)$ is equal to $1$. A set of stations $A$ on the plane is {\em $\delta$-diluted} wrt grid $G_c$, for $\delta\in\NAT\setminus\{0\}$, if for any two stations $v_1, v_2\in A$ with grid coordinates $(i_1,j_1)$ and $(i_2,j_2)$, respectively, it holds true that $(|i_1-i_2|\mod \delta)=0$ and $(|j_1-j_2|\mod \delta)=0$.

  Let $S$ be a general broadcast schedule wrt $N$ of length $T$, let $c, \delta > 0, \delta \in N$. A $\delta$-dilution of $S$ is defined as $(N,\delta)$-gbs $S'$ such that the bit $(t-1)\delta^2 + a \delta + b$ of $S'(v,a,b)$ is equal to $1$ iff the bit $t$ of $S(v)$ is equal to $1$.
	
	As also observed in \cite{JKS-ICALP-13}, any station in a box $C(i,j)$ of the pivotal grid can have communicable neighbors in $20$ boxes. These boxes are called neighboring boxes of box $C(i,j)$. Following \cite{JKS-ICALP-13}, we define the set $DIR \subset [-2,2]^2$ such that $(d_1,d_2) \in DIR$ iff it is possible that boxes with coordinates $(i,j)$ and $(i+d_1,j+d_2)$ can be neighbors. Contrarily, given $(i,j)$ and $(d_1,d_2) \in DIR$, we say box $C(i+d_1,j+d_2)$ is located in direction $(d_1,d_2)$ from box $C(i,j)$.
  
  For each box $C$ in the pivotal grid $G_{\gamma}$, $K_C$ is used to denote the set of nodes which have source-messages ($|K_C| \leq k$). The only information each node $v$ initially has about $K_C$ is whether $v \in K_C$.

	\paragraph{Backbone structure} A \emph{backbone structure}, for a given communication graph $G$, is a subnetwork $H$ which forms a connected dominating set of $G$ with asymptotically the same diameter $D$. The backbone is constructed by selecting a \emph{leader} from each box of $G_{\gamma}$ and a constant number of \emph{helper} nodes to ensure connectivity between neighboring boxes in different directions. Since $H$ has a constant number of nodes in each box, there exists a constant $d$ such that, with $d$-dilution, every node in $H$ can successfully transmit in a constant number of rounds.

\paragraph{Selective families and selectors}
  A family $S=(S_0,\ldots,S_{s-1})$ of subsets of $[N]$ is a {\em $(N,x)$-SSF (Strongly-Selective Family)} of length $s$ if, for every non empty subset $Z$ of $[N]$ s.t.  $|Z|\leq x$ and for every element $z\in Z$, there is a set $S_i$ in $S$ for which $S_i\cap Z=\{z\}$. It is known from ~\cite{ClementiMS01} that there exists $(N,x)$-SSF of size $O(x^2\log N)$ for every $x\leq N$.

  We identify a family of sets $S=(S_0,\ldots,S_{s-1})$ with the broadcast schedule $S'$ such that the $i^{th}$ bit of $S'(v)$ is equal to $1$ iff $v\in S_i$.

  Let $N$, $x$ and $y$ be positive integers so that $y \leq x \leq N$. Let $S$ be a family of subsets of $[N]$. Following \cite{BonisGV03}, we say that $S$ is an \textit{$(N, x, y)$-selector} if for each set $A \subset [N]$ of size $|A| = x$, there are at least $y$ elements in $A$ that can be selected from $A$ by sets in $S$. It is also known that for $y=c x$, where $c\in (0,1)$ is a constant, there is an $(N, x, y)$-selector of size $O(x\log N)$.

\def\GRID{

Given a parameter $c>0$, we define 
a partition of the $2$-dimensional space
into square boxes of size $c\times c$ by the grid $G_c$, in such a way that:
all boxes are aligned with the coordinate axes,
point $(0,0)$ is a grid point,
each box includes its left side without the top
endpoint and its bottom side without the right endpoint and
does not include its right and top sides.
We say that $(i,j)$ are the coordinates
of the box with its bottom left corner located at $(c\cdot i, c\cdot j)$,
for $i,j\in \INT$. A box with coordinates
$(i,j)\in\INT^2$ is denoted $C(i,j)$.
As observed in \cite{DessmarkP07,EmekGKPPS09}, the {\em grid} $G_{r/\sqrt{2}}$
is very useful in the design of the algorithms for geometric radio networks, provided
$r$ is equal to the range of each station.
This follows from the
fact that $r/\sqrt{2}$ is the largest parameter of a grid such that each
station in a box is
in the range of every other station in that box.
In the following, we fix $\gamma=r/\sqrt{2}$, where $r=(1+\eps)^{-1/\alpha}$, and call $G_{\gamma}$
the {\em pivotal grid}. If not stated otherwise, our considerations will
refer to (boxes of) $G_{\gamma}$.

Two boxes $C,C'$ are {\em neighbors} in a network if there are
stations $v\in C$ and $v'\in C'$ such that edge $(v,v')$ belongs to the
communication graph of the network. Boxes $C(i,j)$ and $C'(i',j')$ are {\em adjacent} if
$|i-i'|\leq 1$ and $|j-j'|\leq 1$ (see Figure~\ref{fig:adjacent}).
\begin{figure}
\begin{center}
\epsfig{file=adjacent.eps, scale=0.8}
\end{center}
\caption{If $v,w,z$ are in the range are of $u$, then boxes
containing $v,w,$ and $z$ are neighbors of $C$. The first figure
contains all $20$ boxes which can be neighbors of $C$.
The boxes
$C_1,\ldots,C_8$ are adjacent to $C$.}
\label{fig:adjacent}
\end{figure}%
For a station $v$ located in position $(x,y)$ on the plane we define its {\em grid
coordinates} with respect to the grid $G_c$ as the pair of integers $(i,j)$ such that the point $(x,y)$ is located
in the box $C(i,j)$ of the grid $G_c$ (i.e., $ic\leq x< (i+1)c$ and
$jc\leq y<(j+1)c$).
If not stated otherwise, we will refer to grid coordinates with respect
to the pivotal grid.

A (general) {\em broadcast schedule} $\mathcal{S}$ of length $T$
wrt  $N\in\NAT$ is a mapping
from $[N]$ to binary sequences of length $T$.
A station
with identifier $v\in[N]$ {\em follows}
the schedule $\m{S}$ of length $T$ in a fixed period of time consisting of $T$ rounds,
when
$v$ transmits a message in round $t$ of that period iff
the 
position $t\mod T$ of
$\m{S}(v)$ is equal to $1$.

A {\em geometric broadcast schedule} $\mathcal{S}$ of length $T$
with parameters $N,\delta\in\NAT$, $(N,\delta)$-gbs for short, is a mapping
from $[N]\times [0,\delta-1]^2$ to binary sequences of length $T$.
Let $v\in[N]$ be a station whose grid coordinates
with respect to
the grid $G_c$ are equal to $(i,j)$.
We say that $v$ {\em follows}
$(N,\delta)$-gbs $\m{S}$ 
for the grid $G_c$
in a fixed period of time, 
when $v$ transmits a message in round $t$ of that period iff
the $t$th position of
$\m{S}(v,i\mod \delta,j\mod\delta)$ is equal to $1$.
A set of stations $A$ on the plane is {\em $\delta$-diluted} wrt $G_c$, for $\delta\in\NAT\setminus\{0\}$, if
for any two stations $v_1,v_2\in A$ with grid coordinates $(i_1,j_1)$ and $(i_2,j_2)$, respectively,
the relationships $(|i_1-i_2|\mod \delta)=0$ and $(|j_1-j_2|\mod \delta)=0$ hold.

Let $\m{S}$ be a general broadcast schedule wrt $N$ of length $T$,
let $c>0$ and $\delta>0$, $\delta\in\NAT$.
A $\delta$-dilution of a  $\m{S}$ 
is defined as a $(N,\delta)$-gbs $\m{S}'$ such that the bit $(t-1)\delta^2+a\delta+b$
of $\m{S}'(v,a,b)$ is equal to $1$ iff the bit $t$ of $\m{S}(v)$
is equal to $1$. That is, each round $t$ of $\m{S}$ is
partitioned
into $\delta^2$ rounds
of $\m{S}'$, indexed by pairs $(a,b)\in [0,\delta-1]^2$, such that a station
with grid coordinates $(i,j)$ in $G_c$ is
allowed to send messages only in rounds with index $(i\mod\delta,j\mod\delta)$,
provided schedule $\m{S}$ admits a transmission in its (original) round $t$.
Since we will usually apply dilution to the pivotal grid, it is assumed that all references
to a dilution concern that grid, unless stated otherwise.

Observe that, since ranges of stations are equal to the length
of diagonal of boxes of the pivotal grid, a box $C(i,j)$ can have at most
$20$ neighbors (see Figure~\ref{fig:adjacent}).
We define the set $\DIR\subset[-2,2]^2$ such  that $(d_1,d_2)\in\DIR$ iff
it is possible that boxes with coordinates $(i,j)$ and $(i+d_1,j+d_2)$
can be neighbors.
Given $(i,j)\in\INT^2$ and $(d_1,d_2)\in\DIR$, we say that the box $C(i+d_1,j+d_2)$
is {\em located in direction} $(d_1,d_2)$ from the box $C(i,j)$.

Let $0<x<$
Let $I_1=[i_1,j_1)$, $I_2=[i_2,j_2)$ be segments on a
line, whose endpoints belong to the grid $G_x$.
%
The {\em box-distance} between $I_1$ and $I_2$ with respect go $G_x$ is zero when $I_1\cap I_2\neq\emptyset$,
and it is equal to $\min(|i_1-j_2|/x, |i_2-j_1|/x)$ otherwise. Given two rectangles
$R_1$, $R_2$, whose nodes belong to $G_x$, the box-distance $\distM(R_1,R_2)$ between $R_1$
and $R_2$ is equal to the maximum of the
box-distances between projections
of $R_1$ and $R_2$ on the axes defining the first and the second dimension in the Euclidean
space.
%


Now,
we present a proposition which shows that, if the density of transmitting stations is bounded,
their messages can be heard in a distance related to the density parameter.
\begin{proposition}\labell{prop:lead1}
For each $\alpha> 2$ 
and $\eps>0$,
there exists a 
\tjj{constant $d_{\alpha}$}
such that
the following properties hold.
Assume that a set of $n$ stations $A$ is $d$-diluted wrt the grid $G_x$, where $x=\gamma/c$, $c\in\NAT$, $c>1$
and $d\geq d_{\alpha}$.
Moreover, at most one station from $A$ is located in each box of $G_x$. Then,
if all stations from $A$ transmit simultaneously, each of them
is $\frac{2r}{c}$-successful. Thus, in particular, each station from
a box $C$ of $G_x$ can transmit its message
to all its neighbors located in $C$ and in boxes $C'$ of $G_x$ which are adjacent to $C$.
\end{proposition}
\begin{proof}
Recall that $r=(1+\eps)^{-1/\alpha}$ and $\gamma=r/\sqrt{2}$.
Consider any station $u$ in distance smaller or equal to $\frac{2r}{c}\leq 2\sqrt{2}x<3x$ to a station $v\in A$.
Then, the signal from $v$ received by $u$ is at least
$$\frac1{\left(\frac{2r}{c}\right)^{\alpha}}=\left(\frac{c}{2r}\right)^{\alpha}.$$
Now, we would like to derive an upper bound on interferences caused by stations in $A\setminus\{v\}$
at $u$.
Let $C$ be a box of $G_x$ which contains $v$.
The fact that $A$ is $d$-diluted
wrt $G_x$ implies that the number of boxes containing elements of $A$
which are in box-distance
$id$ from $C$ is at most $8(i+1)$ (see Figure~\ref{fig:gran}).
Moreover, no box in distance $j$ from $C$ such
that ($j\mod d\neq 0$) contains elements of $A$.
\begin{figure}
\begin{center}
\epsfig{file=dist0.eps, scale=0.8}
\end{center}
\caption{Boxes in distance $id$ from $C$ form a frame partitioned into four
rectangles of size $x\times (2id+2)x$. Each of these rectangles contain at most $i+1$
boxes such that any two of them are in box-distance at least $d$.}
\label{fig:gran}
\end{figure}%
Finally, for a station $v\in C$ and a station $w\in C'$ such that $\distM(C,C')=j$,
the inequality $\dist(v,u)\geq jx$ is satisfied.
Note that our goal is {\em not} to evaluate interferences at $v\in C$, but at any station
$u$ such that $\dist(u,v)\leq \frac{2r}c<3x$. Therefore, $u\in C'$ such that $\distM(C,C')<3$, where
$C'$ is a box of $G_x$.
For a fixed $d>3$, the total noise and interferences $I$ caused by
all elements of $A\setminus\{v\}$ at $u$ is at most
$$\cN+\sum_{i=1}^{n}8(i+1)\cdot\frac{1}{(i\bar{d}x)^{\alpha}}$$
where $d\geq\bar{d}\geq d-3$,
since there are at most $8(i+1)$ nonempty boxes in box-distance $i\cdot d$
from the box $C$
in $d$-diluted instance and the box-distance between $C$ and the box $C'$ containing $u$
is at most $2$. Furthermore,
$$I\leq 1+ 8\cdot\left(\frac{1}{\bar{d}x}\right)^{\alpha}\cdot\sum_{i=0}^{n}(i+1)^{1-\alpha}\leq 1+8\left(\frac{c\sqrt{2}}{r \bar{d}}\right)^{\alpha}\sum_{i=1}^{n}i^{1-\alpha}=1+8d_{\alpha}\left(\frac{\sqrt{2}c}{r\bar{d}}\right)^{\alpha}$$
where $d_{\alpha}=\sum_{i=1}^{n}i^{1-\alpha}\leq 1+\zeta(\alpha-1)$, $\zeta$ is the Riemann zeta function and $\cN=1$.
So, 
the signal from $v$ is received at $u$ if the following
inequality is satisfied
\begin{equation}\label{eq:signal}
1+8d_{\alpha}\left(\frac{\sqrt{2}c}{r\bar{d}}\right)^{\alpha}\leq \left(\frac{c}{2r}\right)^{\alpha}
\end{equation}
which is equivalent to
$$\bar{d}\geq 2\sqrt{2}\left(\frac{8d_{\alpha}}{1-(2r/c)^{\alpha}}\right)^{1/\alpha}.$$
Assuming that $c\geq 2$, we have $1-(\frac{2r}{c})^{\alpha}\geq 1-r^{\alpha}$
and therefore (\ref{eq:signal}) is satisfied for each
$\bar{d}\geq 2\sqrt{2}\left(\frac{8}{1-r^{\alpha}}\right)^{1/\alpha}d_{\alpha}$
or $d\geq 3+2\sqrt{2}\left(\frac{8}{1-r^{\alpha}}\right)^{1/\alpha}d_{\alpha}$.


\qed
\end{proof}

The following corollary is a straightforward application of Proposition~\ref{prop:lead1} for $c=2$.
\begin{corollary}\labell{cor:dilsuc}
For each $\alpha>2$ 
there exists
a 
constant $d_{\alpha}$
such that the following
property is satisfied:
Let $A$ be a set of $O(n)$ stations  on the plane which is
$\delta$-diluted wrt the pivotal grid $G_{\gamma}$,
where $\delta\geq d_{\alpha}$ and each box
contains at most one element of $A$. Then, if all elements of $A$ transmit
messages simultaneously in the same round $t$ and no other station is transmitting
a message in $t$, each of them transmits successfully.
\end{corollary}

}

\section{Centralized Setting}
\label{sec:central}
  In the centralized setting, every node has complete knowledge of the coordinates of all other stations. With respect to the vanilla broadcast problem in the centralized setting there are two additional complications: (1) Initially, all nodes are asleep except the nodes in $K_C$ and nodes do not know the members of $K_C$. (2) Secondly, unit size messages. We will show how to address these two issues using two different approaches, which either depend or 
do not depend on the value of the granularity parameter $g$.

  For this setting, single source broadcast can be conducted in $O(D)$ rounds \cite{JKS-ICALP-13}. It is easy to see that $\Omega(D + k)$ is a lower bound on $k$-source broadcast with unit size messages. The main enhancements we develop in these protocols is how the $k$-sources are identified and their messages pipelined on the backbone communication structure with transmission of unit size messages only.

\subsection{Granularity independent algorithm}
\label{broad_cent_nongran}
\subsubsection{Overview}
  Within each box $C$ of the pivotal grid, at most one leader $l(K_C)$ is elected out of the $K_C$ active nodes in at most $k \log \Delta$ rounds. This is achieved by $k$ repetitions of the strongly selective family $(\Delta,c)$-SSF, for appropriately constant $c$, in which only the $k$-source nodes participate.

	It was observed in \cite{JKS-ICALP-13} that irrespective of the number of nodes who transmit in a given round, the closes pair can successfully communicate (i.e. one can hear another). In particular, if nodes execute $(\Delta,c)$-SSF, for $c \geq 2$, then both the nodes belonging to the closest pair of nodes can successfully transmit to each other. The one with greater label value can silence itself. This process can be continued till at most one node remains in each box of the pivotal grid. The leader remaining in the box is the root of an undirected tree, where the nodes belonging to the tree were silenced by their parent node.

  However, this process does not preclude nodes belonging to different boxes of pivotal grid from communicating with each other (and being silenced in the manner described) and belonging to the same tree. If we \textit{dilute} this process \dk{in space}, as described in \cite{JKS-ICALP-13}, and further add the restriction that only node belonging to the same box of the pivotal grid can silence another node, then we can be sure that only nodes belonging to the same box of the pivotal grid belong to any tree created thus.
	
  $O(k)$ repetitions of the above process, is guaranteed to leave at most one active source node in each box of the pivotal grid, irrespective of the initial distribution of source nodes in the network. This node, $l(K_C)$, wakes up all the nodes in box $C$ and coordinates individual transmissions from nodes in $K_C$, using dilution \cite{JKS-ICALP-13}. These messages are gathered on the backbone structure $H$ (which is precomputed in the centralized setting) by the leader $l(C)$, and then pipelined on the entire backbone structure. Finally, the leaders of the boxes in the backbone structure push these messages stored with them in their boxes and every node of the network receives these messages.

  The first stage takes $O(k \lg \Delta)$ rounds. Using pipelining, all $k$ messages reach every node in $H$ in $O(D + k)$ steps. This is followed by $O(k)$ rounds to distribute at most $k$ messages at nodes in $H$ to all nodes in $G$, making the complexity of the algorithm to be $O(D + k \lg \Delta)$ (as last stage just repeats the second stage).

\subsubsection{Connected Dominating Set}
  The node with the least label in each box of pivotal grid is considered to be the leader of the box. For each $(i,j) \in DDL$, let the set of nodes in box $C$ which can have neighbors in box $C(i,j)$ be $S_C^{(i,j)}$. The node with the least label from $S_C^{(i,j)}$, denoted by $s_C^{(i,j)}$, is the $(i,j)$ directional sender - it is a helper node to send a message to $C(i,j)$. Similarly we mark a node to be the $(i,j)$ directional receiver from $C(i,j)$. Let $R_C^{(i,j)}$ be the set of nodes which are connected to $s_C^{(i,j)}$. The node with the least label among $R_C^{(i,j)}$ is the directional receiver, denoted by $r_C^{(i,j)}$.

\alglanguage{pseudocode}
\setlength{\textfloatsep}{0pt}
\begin{algorithm*}[t!]
\small
\caption{Compute-Backbone(Vertex $v$, Graph $G(V,E)$)}
\label{algo:back:1}
\begin{algorithmic}[1]
	\State $C \gets \text{box}(v)$ and $l(C) \gets \min\{u, u\in C\}$
	\Comment{\emph{Leader of C}}
	\ForAll {$(i,j) \in DDL$}
		\State $s_C^{(i,j)} \gets \min \{\text{all nodes with neighbours in } C(i,j) \in C\}$
		\Comment {\emph{Directional sender to $C(i,j)$ in $C$}}
		\State $s_{C(i,j)}^{(-i,-j)} \gets \min \{\text{all neighbours of } C \in C(i,j)\}$
		\Comment {\emph{Directional sender to $C$ in $C(i,j)$}}
		\State $r_C^{(i,j)} \gets \min \{\text{all neighbours of } s_{C(i,j)}^{(-i,-j)} \in C\}$
		\Comment {\emph{Directional receiver from $C(i,j)$ in $C$}}
	\EndFor
\Statex
\end{algorithmic}
\end{algorithm*}

\begin{proposition}
 Algorithm \emph{Compute-Backbone}~efficiently computes the leader in each box as well the helper nodes of the leader in each direction, for the centralized setting.
\end{proposition}

\subsubsection{Message Gathering}
\begin{proposition}
\label{prop:ssf}
  For each $\alpha >2$, there exist constants $d$ and $c$, which depends only on model parameters, satisfying the following property. Let $W$ be a set of stations such that $\min_{u,v\in W, \text{box}(u)=\text{box}(v)}\{dist(u,v)\} = x$ and let $dist(u,v) = x$ for some $u,v \in W, box(u)=box(v)$ and $W$ is $d$-diluted for $d \geq 2$. Then, $v$ can hear the message from $u$ during an execution of a $(N,c)$-SSF on $W$.
\end{proposition}

  Since every node has knowledge of all nodes in the box, we can assign temporary labels to each node from [$\Delta$]. By Proposition \ref{prop:ssf}, at least one pair of nodes in $K_C$ exchange messages in $O(\lg \Delta)$ rounds (referred to as a \emph{step}) using $(\Delta,c)$-SSF. Of the pair, the node with the larger label drops out of the contest while noting the other which remains as its parent.

  It is not known how many messages each node of $K_C$ has (a single node may contain multiple messages), the leader of $K_C$, denoted by 
\dk{$l(K_C)$,} 
must first collect this information. We define a \emph{message tree} $T$ such that $parent(u) = v$ if at some step $u$ won from $v$. The tree $T$ is a min-heap with every node having smaller label than all its children. By definition a node exchanges messages with all of its children in $T$ and so is aware of their labels. We use this to co-ordinate the following round-robin procedure for exploring the tree $T$. $l(K_C)$ requests each of its children node to sequentially transmit their labels and messages. In this manner, $l(K_C)$ explores the structure of $T$ similar to a Breadth First Search. Only the nodes in $K_C$ participate in this protocol.

\alglanguage{pseudocode}
\setlength{\textfloatsep}{0pt}
\begin{algorithm*}[t!]
\small
\caption{Gran-Independent-Collect-Info(Vertex $v$, Graph $G(V,E)$)}
\begin{algorithmic}[1]
\State $state(v)$ is active 
for all nodes in $K_C$, and all other nodes remain inactive
\Comment{\emph{Collects Information about $K_C$}}
	\While {$state(v)$ = active}
		\State Assign unique temporary IDs (TIDs) in [$|C|$] to all elements of $C$
			\State $v$ transmits $m_v$ encoding $v$, according to its $rank(v)$ using $(|C|,c)$-SSF

      \State $state(v) \gets$ inactive \textbf{if} $v$ hears $m_u, u \in \text{box}(v), \text{and } u < v$ \textbf{else} $v$ stores $m_u$
	\EndWhile
\Statex
\end{algorithmic}
\end{algorithm*}

\begin{proposition}
  Algorithm \emph{Gran-Independent-Collect-Info} elects $l(T_C)$ in $O(k \lg \Delta)$ rounds.
\end{proposition}

\alglanguage{pseudocode}
\setlength{\textfloatsep}{0pt}
\begin{algorithm*}[t!]
\small
\caption{Gather-Message(Vertex $v$, Graph $G(V,E)$)}
\label{prot:gath}
\begin{algorithmic}[1]
	\State $q =$ empty queue
	\State $q.$enqueue($l(T_C)$)
	\While {not $q.$empty}
		\State $u \gets q.$dequeue()
		\State $T_C$ requests $u$ to start transmitting.
		\State $u$ transmits each of its children $w$, $l(w)$, and its messages sequentially
		\State $q.$enqueue($u.$children())
	\EndWhile
\Statex
\end{algorithmic}
\end{algorithm*}

\begin{proposition}
  Algorithm \emph{Gather-Message} ensures all the messages in a box (if any) are collected by the leader $l(C)$ of each box $C$ in $O(k)$ rounds.
\end{proposition}

\subsubsection{Message dissemination on backbone structure}
  Each iteration of \emph{Push-Messages} is aimed at each node successfully transmitting a new (first so-far unsent) message to all its neighbors in the backbone. This ensures that in $O(D+K)$ rounds of transmission, all the $k$ messages are received by all the nodes. Note that only nodes in $H$ participate. After the messages have reached all nodes in $H$, messages can be sent trivially to all remaining nodes in $G$ in $O(k)$ transmission rounds.

\alglanguage{pseudocode}
\setlength{\textfloatsep}{0pt}
\begin{algorithm*}[t!]
\small
\caption{Push-Messages(Vertex $v$, Graph $G(V,E)$)}
\begin{algorithmic}[1]
	\State $H_C \gets H \cap C$
	\ForAll{Node $w$ in $H_C$}
		\State Let $rec\_msg\_list \gets$ list of all messages received,
		initially for $l(C)$, $rec\_msg\_list$ is set to list of messages from MSG-GTH.
		\State $w$ transmits first message in $rec\_msg\_list$, so far not transmitted, in round $rank$.
	\EndFor
\Statex
\end{algorithmic}
\end{algorithm*}

\begin{proposition}
  Algorithm \emph{Push-Messages} ensures every node in $H$ successfully transmits a message to all of its neighbours in $O(1)$ rounds.
\end{proposition}

\alglanguage{pseudocode}
\setlength{\textfloatsep}{0pt}
\begin{algorithm*}[t!]
\small
\caption{Central-Gran-Independent-Multicast(Vertex $v$, Graph $G(V,E)$)}
\begin{algorithmic}[1]
	\If{$v \in K_C$}
		\State execute \emph{Gran-Independent-Collect-Info} and \emph{Gather-Message}
	\EndIf
	\If {$v$ wakes up}
		\State Use \emph{Compute-Backbone} to find $H$
		\State execute \emph{Push-Messages} $D+2k$ times.
		\Comment {\emph{All nodes in first $H$ and then in $G$ receive the $k$ messages}}
	\EndIf
\Statex
\end{algorithmic}
\end{algorithm*}

\begin{corollary}
Protocol \emph{Central-Gran-Independent-Multicast} constructs a \emph{backbone structure} and accomplishes multi-Broadcast in the non-spontaneous wake-up setting in $O(D + k \log \Delta)$ .
\end{corollary}

\subsection{Granularity dependent algorithm}
\label{broad_cent_gran}

	The main difference in this algorithm is in the first stage in which the leader $l(K_C)$ is elected and corresponding tree is prepared. The rest of the algorithm is same. We briefly describe an alternate procedure for leader election in $K_C$, which takes $O(\lg g)$ rounds.

  Let $x$ be the smallest distance between two nodes in the network. Then, in $G_x$ there is at most one node in each box of grid $G_x$. Now suppose at an inductive stage grid $G_y$ has the property that each box of the grid has at most one active node, who is the leader in that box. Then, all these leaders transmit their messages in an appropriately (constant) diluted schedule and of the at most four leaders in $G_{2y}$ the one with the least label is chosen the leader, who alone remains active. If we continue this process till the pivotal grid $G_{\gamma}$ is reached, we have ensured that there is at most one single leader $l(K_C)$ left in box $C$ of the pivotal grid. This process is 
further diluted, so that when nodes in one box of the pivotal grid are transmitting, the nodes in nearby 20 boxes are not.

  The \emph{message tree} $T$, as in the previous section, is used to coordinate by the leader $l(K_C)$ to coordinate in $O(k)$ rounds, in which each node of $K_C$ gets a separate round for itself to transmit its message. These messages are gathered by the leader of the box in backbone structure $H$. Finally, pipelined transmission of the gathered messages happen over the backbone structure in $O(D)$ rounds. Respective leaders of the boxes disperse these collected messages to the nodes in their box in another $O(k)$ rounds.

\alglanguage{pseudocode}
\setlength{\textfloatsep}{0pt}
\begin{algorithm*}[t!]
\small
\caption{Gran-Dep-Collect-Info(Vertex $v$, Graph $G(V,E)$)}
\begin{algorithmic}[1]
	\State $h \gets \min_{i \in \mathbb{N}}(2^i|2^i \geq g)$ and $y \gets r/h$
	\State $state(v)$ is active initially for all nodes in $K_C$, and all other nodes remain inactive.
	\For{$i = 1,2,\dots,\log h$}
		\State Each active node in $G_{2y}$ transmits sequentially encoding its label.
		\Comment {\emph{There are only four such nodes}}
        \State $state(v) \gets$ inactive \textbf{if} $v$ hears $m_u,  u \in G_{2y}, \text{and } u < v$ \textbf{else} $v$ stores $m_u$
	\EndFor
\Statex
\end{algorithmic}
\end{algorithm*}

\begin{proposition}
  Algorithm \emph{Gran-Dep-Collect-Info}, $\delta$-diluted with respect to pivotal for $\delta = 5$, elects $l(T_C)$ for each subset $K_C$ in $O(\lg g)$ rounds.
\end{proposition}
\begin{corollary}
  Protocol \emph{Central-Gran-Dependent-Multicast} which replaces \emph{Gran-Independent-Collect-Info} with \emph{Gran-Dep-Collect-Info} in \emph{Central-Gran-Independent-Multicast} constructs a \emph{backbone structure} and accomplishes multi-Broadcast in the non-spontaneous wake-up setting in $O(D + k + \lg g)$ rounds.
\end{corollary}
\section{Networks with knowledge of \\neighbors and their coordinates}
\label{broad_loc_spont}
  
\dk{In the setting where nodes have knowledge about their and neighbors' coordinates,
we present an algorithm that works as follows.}
First all the active nodes $K_C$ execute a selective family, where-in every node which receives some message shuts itself off for the time. This results in each box of the pivotal grid having at most one active node that belongs to $K_C$. Now, Algorithm Gen-Inter-Box-Broadcast from \cite{JKS-ICALP-13} is executed $D$ times. It is unaffected by the fact that there are multiple boxes in pivotal grid in which there are active nodes at the start of the algorithm. This results in the following: (a) All nodes are awake in the network. (b) There is a local leader elected in each box of the pivotal grid. (c) There are directional senders that are elected for communicating in each connected direction in 
\dk{$DIR$.}

	From this stage we prepare the remaining communication infrastructure for the graph as follows: Each of the directional senders chooses one of the stations from its set of neighbors, that belongs to the corresponding box in that direction and declares it the directional receiver for that box. Thus, if a message is to be sent from one box of the grid to another neighboring box, it is accomplished as follows. The local leader first transmits the message, followed by the transmission of the message by the directional sender. It is then received by the directional receiver who then transmits it once and is received by the local leader of that box.

  This communication infrastructure is used as follows: The local leader collects all the $\leq k$ source-messages from its box in $k \lg \Delta$ time and puts them in a stack. Then, the above communication procedure for forwarding the message from one box to the adjacent box is executed for each direction. Thus, all adjacent boxes receive the message. This process can be executed concurrently in case there are other boxes carrying some source message with the help of appropriate dilution. If a message that has been transmitted by local leader to its adjacent boxes, is received again from another direction, then it is ignored. If there are multiple messages that arrive at the local leader in a sequence of rounds, then they are all stored in a  stack along with previously unsent messages. When the local leader finds a free set of rounds that were meant for it to transmit messages to neighboring boxes, then it pops a new message from the stack and executes the above transmission procedure.

  If this procedure is executed $D+k$ times, then it is easy to see that all the $k$ source-messages reach all remaining nodes in the network.

\begin{proposition}
\label{prop_elec}
  Algorithm \emph{Gen-Inter-Box-Broadcast} (\cite{JKS-ICALP-13}) works in time $O(\lg ^2 n)$ and selects a leader from a set of nodes with local knowledge.
\end{proposition}

  We assume that every node has knowledge of the labels as well as positions of all of its neighbors. Note that the algorithm described in the previous section requires complete knowledge of the topology for only building the communication backbone (a connected dominating set) $H$, after which local knowledge is sufficient i.e. $v$ needs to have knowledge of nodes $u \in G, u \in C$ and $u \in H, u \in C$, where $C = \text{box}(v)$.

  After all nodes in a box have been woken up, we create a backbone using Proposition \ref{prop_elec}. Thus, on each box being activated, we elect leader in the box and compute the nodes which belong to the backbone, and wake-up all the nodes in the neighboring boxes. Thus, in $O(D)$ repetitions, the backbone structure of the entire network is computed. This is followed by gathering and dissemination of the $k$ messages - first in the backbone and then to the entire network in $O(D + k \log n)$ rounds.

\subsection{Connected Dominating Set}
 The node with the least label in each box is considered to be the leader of the box. For each $(i,j) \in DDL$, let the set of nodes in box $C$ which can have neighbors in box $C(i,j)$ be $S_C^{(i,j)}$. We have to elect a leader from each of these $S_C^{(i,j)}$, denoted by $s_C^{(i,j)} = l(S_C^{(i,j)})$ which is the helper node to send a message to $C(i,j)$. This is followed by a round where $s_C^{(i,j)}$ transmits to inform all neighbors that it is the designated sender. Similarly we mark a node to be the designated receiver from $C(i,j)$. Let $R_C^{(i,j)}$ be the set of nodes which are connected to $s_C^{(i,j)}$. We elect a leader out of them, denoted by $r_C^{(i,j)} = l(R_C^{(i,j)})$ and notify its neighbors.
\alglanguage{pseudocode}
\setlength{\textfloatsep}{0pt}
\begin{algorithm*}[t!]
\small
\caption{Local-Leader-Elect(Vertex $v$)}
\label{algo:back}
\begin{algorithmic}[1]
   	\State $C \gets \text{box}(v)$
	\State $l(C) \gets \min\{u, u\in C\}$
	\Comment{\emph{Leader of C}}
	\ForAll {$(i,j) \in DDL$}
		\State $s_C^{(i,j)} \gets \text{\emph{Gen-Inter-Box-Broadcast}} \{\text{all nodes with neighbours in } C(i,j) \in C\}$
		\Comment {\emph{Designated sender to $C(i,j)$ in $C$}}
		\State $u \gets s_C^{(i,j)}$, transmits $m_u$ encoding $u$
		\Comment {\emph{The sender notifies its neighbours}}
		\State $r_C(i,j)^{(-i,-j)} \gets \min\{\text{all neighbours of } u \in C(i,j)\}$
		\Comment {\emph{Designated receiver from $C$ in $C(i,j)$}}
		\State $u$ transmits $m_r$ encoding $r_C(i,j)^{(-i,-j)}$, waking it up
		\Comment {\emph{The sender notifies the receiver}}
	\EndFor
\Statex
\end{algorithmic}
\end{algorithm*}

\begin{proposition}
  Protocol \emph{Local-Leader-Elect} has round complexity of $O(\lg ^2 n)$ rounds and elects  leader and helper nodes (for communicating with adjacent boxes) in each box, given that all nodes are aware of their neighborhood.
\end{proposition}

\alglanguage{pseudocode}
\setlength{\textfloatsep}{0pt}
\begin{algorithm*}[t!]
\small
\caption{Local-Multicast(Vertex $v$)}
\label{prot:diss}
\begin{algorithmic}[1]
	\If{nodes $S_C$ of a box wake up}
		\State execute \emph{Gen-Inter-Box-Broadcast} to elect leader of $S_C$ ($l(S_C)$)
		\State $l(S_C)$ transmits, waking up all nodes in $C$
		\State execute \emph{Local-Leader-Elect}
	\EndIf
	\If{backbone ($H$) has been constructed}
		\State execute \emph{Gather-Message}
        \State execute \emph{Push-Messages} $D+2k$ times.
		\Comment {\emph{All nodes in first $H$ and then in $G$ receive the $k$ messages}}
	\EndIf
\Statex
\end{algorithmic}
\end{algorithm*}

\begin{corollary}
Protocol \emph{Local-Multicast} constructs a \emph{backbone structure} and accomplishes multi-cast in the non-spontaneous wake-up setting in $O(D \lg^2 n + k \log \Delta)$ rounds, given that each node is aware of its neighborhood.
\end{corollary}
\section{Nodes with Knowledge of own Coordinates only}
\label{s:only-own-coord}
  In this setting, initially each node knows only its own coordinates, label, $n$, $N$ and whether it belongs to subset $K_C$ or not.

  The protocol proceeds in three phases. In first phase, the number of active nodes in any box are reduced to at most $1$, using the selective family based method used in previous section, in $\S$ \ref{broad_cent_nongran}. In the second phase, the active nodes wake-up all the other nodes. Among the active nodes, a leader is elected who coordinates a round robin where each node in the box transmits in a separate round (done with appropriate dilution). Thus, all nodes learn their neighborhood. Using this extra knowledge, a backbone structure is easy to construct along the same lines as done in previous section. In the third phase, the source messages are gathered from the $k$-sources on some node of the backbone structure and then pipelined to reach the rest of the backbone structure. Finally, these messages are pushed to the network by the local leaders in each box.

  In the second phase, two threads are executed simultaneously using time multiplexing (as is done in \cite{JKS-ICALP-13}). In one thread, leader election is conducted by repeated execution of $(N,c)$-SSF for an appropriately sized constant $c$. A \textit{message tree} $T$ is maintained as in the previous sections by the leader. In the second thread, the current leader(s) in the box execute a round-robin like protocol in which they get nodes belonging to the \emph{message tree} in their box to transmit one by one, using a BFS like procedure. A round robin conducted in one box may be disrupted by round robins conducted in far off boxes in which multiple nodes are transmitting at the same time because they have not elected a single leader amongst them. However as shown in $\S 4$ of \cite{JKS-ICALP-13}, this process is successful in waking up every node and successfully conducting round-robin in every box in $O(n \lg N)$ rounds, because progress is always made on the strongly selective family thread which occurs in rounds dedicated for it.

  The second phase successfully completes in $O(n \lg N)$ rounds and has a property that every node in every box has successfully transmitted once without any interference from far off nodes. Thus, all nodes update their neighborhood by the end of the second phase. The local leaders merely demarcate their directional senders for communicating in each direction, who in turn announce their directional receivers for box in that direction. This sets up the communication infrastructure which can now be employed for the $k$-source broadcast in a pipelined manner as done in previous section.

\subsection{Leader Election}
  When all nodes of a box have been activated, first a leader is chosen, who coordinates the modified round-robin procedure. This election is similar to the manner in which leader is elected in \emph{Gran-Independent-Collect-Info} in $\S$\ref{broad_cent_nongran}. Recall that with the execution of a $\log N$ size selector $(N,c)$-SSF, the nodes which are closest to each other successfully communicate with each other. We do not know when exactly this thread will accomplish the leader election process. To overcome this, algorithm is executed for $O(n \lg N)$ rounds, which is an upper bound on its running time. All transmissions in the leader election are done only in odd rounds to ensure it does not affect other transmissions.

\alglanguage{pseudocode}
\setlength{\textfloatsep}{0pt}
\begin{algorithm*}[t!]
\small
\caption{Thread1(Vertex $v$, Time $t$)}
\begin{algorithmic}[1]
	\If{$state(v) = $ active}
		\State $(i,j) \gets$ box-coordinates of $\text{box}(v)$
		\State let $m_v$ encode $v$ and $(i \mod 10, j \mod 10)$
		\State $v$ transmits $m_v$ according to its label using $(N,c)$-SSF in rounds $t \mod 2 = 1$
		\If{$v$ hears $m_u, u\in \text{box}(v)$}
			\State $state(v) \gets$ inactive \textbf{if} $u < v$ \textbf{else} $G(v) \gets G(v) \cup \{m_u\}$
        \EndIf	
	\EndIf
\Statex
\end{algorithmic}
\end{algorithm*}

\begin{proposition}
Thread1 elects $l(T_C)$ in $O(|C| \log n)$ rounds, when each node knows its label and coordinates.
\end{proposition}

\subsection{Round~Robin}
  In this thread, all the nodes in the box transmit one-by one, ensuring at least one message from each node is passed to all its neighbors. This round also ensures all the neighboring boxes are woken-up. To prevent the signals from leader elections happening far-off to interfere, we execute this thread entirely in the even rounds. After completion of \emph{Elec-Lead}, each node is aware of its neighborhood, and protocols which take advantage of local knowledge can now be used to complete multi-casting.
  
\alglanguage{pseudocode}
\setlength{\textfloatsep}{0pt}
\begin{algorithm*}[t!]
\small
\caption{Thread2(Vertex $v$, Time $t$)}
\begin{algorithmic}[1]
	\State $leader \gets l(G(v))$, $q \gets \{leader\}$
	\While{$q$ is not empty}
		\If{$t \mod 2 = 0$}
			\State $u \gets q.$dequeue()
			\State $leader$ requests $u$ to transmit
			\State $u$ transmits $\{w| w \in G(u)\}$ and the first new message it has
			\State $q.$enqueue($u.children$)
			\Comment {\emph{as inferred from $G(u)$}}
		\EndIf
	\EndWhile
\Statex
\end{algorithmic}
\end{algorithm*}

\begin{proposition}
  Running Thread1 and Thread2 in parallel ensures in $O(n \lg N)$ rounds, each box has a unique leader who knows the labels of all the nodes in the box and whose label is known by all in the box, and every node is aware of the labels and box-coordinates of its neighbors.
\end{proposition}

\subsection{Phase~3}
  At the end of Phase $2$, each node has transmitted successfully at least once, and so every node is aware of its local neighborhood. The local knowledge is used to construct the backbone communication structure, which is then used to disseminate the messages throughout the network as in the previous section.

\alglanguage{pseudocode}
\setlength{\textfloatsep}{0pt}
\begin{algorithm*}[t!]
\small
\caption{Construct-Backbone(Vertex $v$)}
\begin{algorithmic}[1]
	\ForAll{$(i,j) \in DDL$}
		\State execute Thread2 where each node $u$ transmits whether it has a neighbour in $C(i,j)$ sequentially
		\State $H_C \gets \{l(C)\}$
		\State $s_C^{(i,j)} \gets \min\{\text{all nodes  $\in C$ with neighbours in } C(i,j)\}$ transmits $u$
		\Comment {\emph{Designated sender to $C(i,j)$ in $C$}}
		\State $H_C \gets H_C \cup \{s_C^{(i,j)}\}$
		\Comment {\emph{The sender notifies its neighbours}}
		\State $s_C^{(i,j)}$ transmits $r_C(i,j)^{(-i,-j)} \gets \min\{\text{all neighbours of } s_C^{(i,j)} \in C(i,j)\}$
		\Comment {\emph{Designated receiver from $C$ in $C(i,j)$}}
		\State $H_C(i,j) \gets H_C(i,j) \cup \{r_C(i,j)^{(-i,-j)}\}$
		\Comment {\emph{The sender notifies the receiver}}
	\EndFor
\Statex
\end{algorithmic}
\end{algorithm*}

\begin{proposition}
  Protocol \emph{Construct-Backbone} gives us the set $H_C$ in $O(n)$ rounds, if every node is aware of the labels and box-coordinates of its neighbours, and a leader has been elected for each box.
\end{proposition}

\alglanguage{pseudocode}
\setlength{\textfloatsep}{0pt}
\begin{algorithm*}[t!]
\small
\caption{General-Multicast(Vertex $v$)}
\begin{algorithmic}[1]
	\If{$v \in K_C$}
	\Comment \emph{Phase 1}
		\State execute \emph{Gran-Independent-Collect-Info}
		\State execute \emph{Gather-Message}
	\EndIf
	\State run Thread1 and Thread2 in parallel for $O(n \log N)$ rounds
	\Comment \emph{Phase 2}
	\State execute \emph{Construct-Backbone}
	\Comment \emph{Phase 3}
	\State execute \emph{Push-Messages} $D+2k$ times.
	\Comment {\emph{All nodes in first $H$ and then in $G$ receive the $k$ messages}}
\Statex
\end{algorithmic}
\end{algorithm*}

\begin{corollary}
  Protocol \emph{General-Multicast} accomplishes multi-Broadcast in the non-spontaneous wake-up setting in $O((n + k) \lg N)$ rounds, given that each node is aware of its label, coordinates, $k$, $N$, and $n$.
\end{corollary}
\section{Nodes with knowledge of only immediate neighborhood}
\label{a:1}
  In the setting when nodes do not know coordinates, even their own, it is impossible to
apply coordinate-based techniques such as grid partition or dilution, which is heavily used
in all protocols we have for the SINR model. We show that the knowledge of the ids of reachable neighbors, along with the general understanding that the nodes are embedded in a 2-dimensional euclidean plane, is sufficient to perform multi-broadcast in time $O((n+k) \lg n)$.

  Our solution is based on several new ideas developed for the SINR model, the main of which are: 
(a) The game of tokens, which are passed around in the network and compete whenever they meet in the same box of the pivotal grid (b) The specific way of network traversal and how the spanning tree is defined on that basis which allows to propagate rumors quickly along the tree.

  The main challenge in achieving the sought time performance $O((n+k)\lg n)$ when traversing and spanning a multi-broadcast-suitable backbone tree in a distributed way is in handling the unpredictable interference from other transmitting nodes. More precisely, other token holders do not know location and number of other transmitters, and thus cannot easily predict the amount of interference, unless scheduling long intervals of silence which in turn bursts time performance.

  Thus, simple graph-based searches and games developed e.g., in the context of radio networks are
not directly applicable.

  The high-level idea of the protocol is as follows. It consists of two algorithms: BTD\_Traversals
 and BTD\_MB. In the first algorithm, activated nodes (i.e., nodes with rumors) issue their tokens (consisting of their own id), and then traverse the network and compete with other tokens until only one token dominates and spans a tree on the network, with very specific properties (to be defined later). In the second algorithm, the internal nodes of this spanning tree are used as a backbone structure, though it may not be literarily so, to propagate rumors to all other internal nodes of the tree and their neighbors (i.e., leaves). 

  There are several technical challenges that have to be overcome in implementing the above high-level protocol description. The main one is around the choice of the right distributed multi-traversal method. It has to be quick, allow fast resolution of conflicts between nearby tokens and construct a backbone tree that allows smooth parallel propagation of rumors (in the second algorithm) despite encountering the inevitable interference from nodes both inside and outside their transmission range. The second technical challenge is to bound the number of directly competing tokens (i.e., tokens that visit same node), in order to be able to resolve such conflicts quickly and eliminate all but one.

  We start by addressing the second challenge.

\paragraph{Resolving conflicts of tokens}
Suppose that there is a set of nodes $X$ such that each of these nodes has a status of token holder
and id associated with this status (corresponding to id of some node). Ids are pairwise different.
Each node $v$ in $X$ intends to pass its token to some of its neighbors $\alpha(v)$, called destination. For the purpose of argumentation, consider the pivotal grid $G_\gamma$ and assume that in every box of the grid there is at most one node from $X$. The goal is to design a distributed procedure, that works in $y=O(\lg n)$ rounds guaranteeing the following properties at the end of round $y$: 

\begin{description}
\setlength{\leftskip}{-0.35cm}
\item[(i)] for each token, there is at most one node having status of the holder of this token,
and if there is one, it is the destination node of this token;

\item[(ii)] in each box of the pivotal grid there is at most one node with token holder status;

\item[(iii)] the smallest token is delivered and stored in its destination, i.e., the destination of the smallest token has the status of this token holder and no one else does.
\end{description}
%
  Note that we do not want to guarantee that all tokens will be successfully passed to their
destinations, but only that some of them will be passed while other could be dropped, as long
as conditions (i)-(iii) hold.

  To solve the above problem, consider the following procedure Smallest\_Token$(X)$. Fix a $(N,c)$-SSF for sufficiently large constant $c$, to be defined later. In the first part of the procedure, all nodes in $X$ execute transmissions according to the $(N,c)$-SSF, transmitting their tokens together with id of the destination. In the second part, all destinations that receive tokens addressed to them, pick the smallest such token and transmit it using the same transmission schedule
as in the first part, i.e., based on the $(N,c)$-SSF. After that, each destination node that has not received any smaller token during part two, takes the smallest one received by it in part one, and changes its status to the token holder of the smallest token received in part one. All other nodes become no token holders.

\begin{lemma}
\label{l:many-tokens}
There is a constant $c>0$, such that if the intersection of each box of the pivotal grid $G_\gamma$ with set $X$ is of size one then the above procedure Smallest\_Token$(X)$ based on any $(N,c)$-SSF executed on set $X$ and their tokens computes a new set of token holders satisfying properties (i)-(iii).
\end{lemma}

\begin{proof}
  Note that by assumption that $X$ has at most singular intersection with each box of the pivotal grid, we deduct that each node in $X$ has at most a constant number of neighbors in set $X$. Even more, in any ball of radius $O(r)$, there is a constant number of nodes in $X$ (depending on radius and $r$). 

  Let $c'>1$  be a sufficiently large constant and consider a node $v\in X$ and the ball of radius $c' r$ centred at it. Let $c^*$ be the constant upper bound on the number of nodes from $X$ in any ball of radius $c' r$; note that $c^*$ is also a constant.

  For each node $w\in X$ being a neighbor of $v$, and thus located in the ball, there is at least one round during the execution of any $(N,c^*)$-SSF such that $w$ transmits and no other node in the intersection of the ball with $X$ transmits. By the definition of $G_\gamma$ (i.e., the fact that $\gamma=\Theta(r)$) and the assumption that $X$ has at most singular intersection with every box of $G_\gamma$, the total interference of nodes outside the ball, measured at $v$, is bounded, and moreover, can be made sufficiently small (i.e., smaller than $\eps\beta\cN$) by taking sufficiently large parameter $c'$ defining the ball radius $c' r$. Therefore, even if all nodes in $X$ located outside the ball transmit in the considered round, their interference at node $v$ is smaller than $\eps\beta \cN$ (for sufficiently large constant $c'$), and thus  the SINR value of the signal from node $w$ measured at node $v$ is higher than the threshold $\beta$, and consequently the message from $w$ is successfully received by $v$.

  The above argument implies that after a single execution of $(N,c^*)$-SSF each destination receives a token (or all tokens, in general) addressed to it (among possibly other tokens). From the right ones, it chooses the smallest and repeats transmission schedule according to the $(N,c)$-SSF, for constant $c$ to be defined later. Similar arguments as in the first two paragraph of the proof can be proved --- the only difference is that there might be at most $c^*$ transmitters, instead of one, in a single box of the pivotal grid. Therefore, there is a slightly higher constant, call it $c$, upper bounding the number of nodes transmitting in part two located in a ball of radius $c' r$. Same argument, as for the first part, justifies that any $(N,c)$-SSF is enough that the centre of a ball successfully receives a message from any participating neighbor.

  Finally, note that in part one we used any $(N,c^*)$-SSF, while in part two any $(N,c)$-SSF. However, since we assumed $c>c^*$, any $(N,c)$-SSF is also $(N,c^*)$-SSF. Hence for the purpose of the whole procedure, we can pick any $(N,c)$-SSF and apply it to both parts.
\end{proof}

  Since there exists $(N,c)$-SSF of length $O(\lg n)$, by~\cite{ClementiMS01}, we can plug it
in procedure Smallest\_Token$(X)$, for sufficiently large constant $c$ that satisfies Lemma~\ref{l:many-tokens}, to obtain the following result.

\begin{corollary}
\label{c:many-tokens}
  Procedure Smallest\_Token$(X)$ could be implemented in such a way that in time $O(\lg n)$ it produces a set of token holders satisfying conditions (i)-(iii), provided original set of token holders $X$ satisfied condition (ii).
\end{corollary}

\paragraph{BTD traversal and spanning tree in the SINR model}

  The concept of \BTD\ traversal and spanning tree was introduced in~\cite{CKPR-ICALP-11} in the context of radio network model. We briefly describe the high-level idea, as given in~\cite{CKPR-ICALP-11}, and then we design its fast implementation in the SINR model. To the best of our knowledge, this is the first time a \BTD\ tree is efficiently spanned and used as a backbone for distributed communication under the SINR model.

  Assume that in the beginning one node, called a root, has a token. The token, which contains the id of the root who initiated it, is propagated along the network, starting from the root, by sending it to a neighboring node using  so called a {\em token message}. The goal of the token is to visit all the nodes and return to the root.

  The concept of a \BTD\ traversal is similar to that of DFS.  The difference is that in an instance of passing the token from the token's holder to its unvisited neighbor only one edge is added to the DFS tree (i.e., the edge connecting both these nodes), while in the construction of \BTD\ all neighbors of the token's holder that are outside of the current \BTD\ tree get connected by an edge to the token's holder, and thus to the tree. This action is done by marking adjacent nodes: nodes use {\em checking/reply messages} to confirm which node marked which.

  In order to implement the BTD traversal in the SINR model, we cannot rely on the implementation in the related radio model~\cite{CKPR-ICALP-11}, as it was based on procedure Echo that could emulate collision detection capability at the station keeping the token. Such procedure cannot be efficiently implemented in the SINR model, as it requires knowledge of a neighbor that is of largest distance from the node with the token (recall that in the considered setting, we assume that knowledge of coordinates is not available).

  \BTD\ search  from a designated root can be accomplished in the SINR model by the following 
distributed procedure  \SaSmod. Nodes have  status either \emph{visited} or \emph{unvisited}, depending whether or not they received the token by the current round.  Once a node becomes {\em visited}, its status never changes.  Nodes with status {\em unvisited} have sub-status {\em marked} or {\em unmarked}, depending on whether they have already been approached by their future parent in the partially constructed \BTD\ tree. 

  This is done by sending control messages of type {\em checking} or {\em reply}. 

  In the beginning all nodes except the root are {\em unvisited} and {\em unmarked}, and the root holds the token. Additionally, all nodes $v$ initialize their local list $L_v$ of {\em unmarked} neighbors to all their neighbors except the root (if applicable). These lists allow nodes to record information about marking progress gradually, as sending large messages (e.g., containing whole neighborhood or all known marked neighbors) is not allowed in our model. The computed \BTD\ tree is stored locally in (self-explained) variables $parent(v)$ and lists $Child_v$, initially set to nil. (Technically, at some point the lists $Child_v$ become empty after sending the token to all elements on the lists, but each node can easily archive them - we skip this issue to focus on important aspects of the exploration protocol.)

  The following distributed procedure \SaSmod\ is repeated until termination.

\begin{description}
\setlength{\leftskip}{-1cm}
\setlength{\itemindent}{-0.4cm}
\item{After receiving a token for the first time} (i.e., by an {\em unvisited} node):
After receiving a token message $\langle token,\tau,v,w \rangle$ for the first time from some node $v$,  node $w$, different from the root,  changes its status to {\em visited} and keeps it till the end of the algorithm.  It sets up its parent in the partially created \BTD\ tree: $parent(w)\gets v$.
  Then it sends control messages $\langle check,\tau,w,z \rangle$, one per each node $z$ in $L_w$, 
containing the token id $\tau$ (but not the token itself),  $w$ and the id of the corresponding node $z$ in $L_v$. After each such message $\langle check,\tau,w,z \rangle$, node $w$ removes $z$ from $L_w$ and listens one round for reply; if it hears message $\langle reply,\tau,z,w \rangle$,
it adds $z$ to its list $Child_w$, otherwise it does nothing.
  Then it sends a token message $\langle token,\tau,w,z \rangle$ to the first node $z$ on list $Child_w$ and removes $z$ from $Child_w$.

  When a node $z\ne w$ receives token message $\langle token,\tau,v,w \rangle$, it does nothing.

\item{After receiving a token by a {\em visited} node: } After receiving a token message $\langle token,\tau,v,w \rangle$ by a {\em visited}  node $w$ from node $v$, node $w$ sends token message $\langle token,\tau,w,z \rangle$ to the first node $z$ on $Child_w$, provided the list is not empty.
If the list is empty, node $w$ sends token message $\langle token,\tau,w,parent(w) \rangle$ to its parent node $parent(w)$, if $parent(w)\ne nil$, or finishes the algorithm otherwise (i.e., if $w$ is the root).

\item{After receiving a checking message: }
  After receiving a checking message $\langle check,\tau,v,w \rangle$ from $v$, node $w$: changes its sub-status to {\em marked} and sends a reply message $\langle reply,\tau,w,v \rangle$ in the next round.

  When an {\em unvisited} node $z\ne w$ receives checking message $\langle check,\tau,v,w \rangle$, it records that $w$ is being {\em marked} by removing it from its list $L_z$, unless $w$ is not on $L_z$.

  When a {\em visited} node $z\ne w$ receives checking message $\langle check,\tau,v,w \rangle$, it does nothing (this is only a safety case, which should not occur in valid executions of the algorithm).

\item{After receiving a reply message: }
  When node $w$ receives a reply message $\langle reply,\tau,v,w \rangle$ from $v$, it adds $v$ to its list $Child_w$ (this might happen only when $w$ sent checking message to $v$ during the first time it hold the token, c.f., the specification after receiving token message for the first time). 

  Upon receiving $\langle reply,\tau,v,z \rangle$, for any $z\ne w$, node $w$ removes $v$ from its list $L_w$ if $v$ was on the list.
\end{description}

\begin{lemma}
\label{l:SaSmod}
  Procedure  \SaSmod\  performs a \BTD\ search on the {whole} network and spans a \BTD\ tree rooted at the initiating node in $O(n)$ rounds. The \BTD\ tree is stored locally in the following sense: each node knows its parent and children (i.e., nodes to which the given node sent a token message).
\end{lemma}

\begin{proof}
  Note that in each round exactly one transmission occurs. The number of rounds in which token message is transmitted is at most $2n$, as it traverses the network along some tree based on the status (visited/unvisited). In order to argue that all nodes are visited,  observe that in any round every node that has status {\em unvisited} at this round is on list $L$ of some of its neighbor (on all neighbors' lists if it is unmarked, and at least on the list of the neighbour who marked him afterwards). Also, the token traversing recursively cannot finally leave any node
without emptying its list $L$ first, which in turn decreases only when some nodes in it are being
marked or the token is sent to them. Therefore, by simple recursive argument, the token can finish exploration only when all lists $L$ are empty, which means that all nodes  have status {\em visited} and thus have been already visited.

  To conclude the first part of the lemma and prove $O(n)$ time complexity, it is enough to prove that the number of rounds in which checking or reply message is sent is proportional to the number of nodes too. To see this, consider a node $w$ and all checking and reply messages sent by this node and to this node. First consider checking messages sent to $w$ and replies to them sent by $w$. After the first such message is sent, it is received by $w$ and a reply message is sent by $w$. After that all neighbors of $w$ remember that $w$ has been marked and will not send any control message to $w$. Hence there are $O(n)$ such checking and replying messages. Checking messages sent by $w$ to some processes $z$ and replies received from them can be accounted to $z$, as above, which does not change the total asymptotic number of checking  and replying messages.

  The second part of the lemma follows directly from the algorithm description.
\end{proof}

  The following structural property of the spanned \BTD\ tree will be useful in the analysis of the main multi-broadcast analysis.

\begin{lemma}
\label{l:SaSmod-prop}
  In each box of the pivotal grid there are at most $37$ internal (i.e., non-leaf) nodes of the \BTD\ tree outputted by the procedure \SaSmod.
\end{lemma}

\begin{proof}
  We first prove that for each node $w$ there are at most $36$ of its neighbors who are internal nodes in the subtree rooted at $w$ (here we use the fact that \SaSmod\ produces a \BTD\ tree, c.f., Lemma~\ref{l:SaSmod}). Otherwise, there would be node $w$ with more than $36$ neighbors who are internal nodes in the subtree rooted at $w$. It means that each of these neighbors found some {\em unmarked} neighbour in its own neighborhood. Hence, there are more than $36$ such {\em unmarked}
neighbors altogether, all of which are at most distance $2r$ from $w$. By geometric property
of Euclidean plane, there are two of them, say $v_1$ and $v_2$, which are at distance smaller than $r$ from each other. One of them, assume w.l.o.g. $v_1$, was visited before the other ($v_2$), 
hence by the time the token left $v_1$ for the first time, call this time $t$, node $v_2$ had to be marked. Recall that, by definition of $v_2$, it is marked by some neighbor $v_2^*$ of node $w$
who ends up to be an internal node of the subtree rooted at $w$.

  By the recursive nature of the algorithm and the fact that $v_1$ had been marked by $v_1^*$ before $v_2$ was by $v_2^*$, node $v_2^*$ had not been visited before the token finally left $v_1^*$
after finishing exploration initiated at $v_1^*$ (i.e., after spanning the subtree initiated at $v_1^*$ and containing $v_1$). Therefore node $v_2^*$ had not been visited before time $t$. Hence,
when the token visited $v_2^*$ for the first time, which is after $t$, $v_2$ had been already marked. This contradicts the definition of $v_2$ as the node marked by $v_2^*$.

  Now we prove the statement of the lemma. If there were more than $37$ nodes in a single box
of the pivotal grid, consider the one of them, call it $w$, which was visited first. All other at least $37$ nodes are neighbors of this node, so by the recursive nature of the algorithm they have to be in the subtree rooted at $w$. This however violates the property proved in the previous paragraphs.
\end{proof}

\paragraph{Algorithm BTD\_Traversals}
This algorithm handles the situation of building a single \BTD\ tree in case there are many tokens in the beginning. It uses \SaSmod\ and Smallest\_Token as its subroutines.  It consists of three stages. 

\begin{description}
\setlength{\leftskip}{-1cm}
\setlength{\itemindent}{-0.4cm}
\item{Stage 1: Elimination of contending neighbors.}
  Nodes having rumors execute subsequently $(N,(2/3)^i n, (2/3)^i n/2)$-selectors, for $i=1,\ldots,\lg_{3/2} n$. There exist such selectors of length $O((2/3)^i n \lg n)$, for any $1 \le i \le \lg_{3/2} n$, c.f.,~\cite{BonisGV03}. Whenever a participating node hears a message from a node with smaller id, it becomes idle by the end of this stage. Nodes that initially do not have tokens keep idle during this stage.

  Each nodes that survives (i.e., keeps executing selectors) till the end of this stage, issues a token with id equal to its own id.

\item{Stage 2: Token elimination and spanning a \BTD\ tree.} Each node holding a token at the beginning of this stage initiates the modified \SaSmod\ procedure for the token. The modification has two aspects:
\begin{itemize}
\setlength{\leftskip}{-1.2cm}
\item Each round of the original \SaSmod\ procedure is simulated by execution of procedure Smallest\_Token$(X)$ by the nodes who intends to send a message to some other node; each such message is associated with some token; procedure Smallest\_Token$(X)$ is instantiated using $(N,c)$-SSF that satisfies the statement of Corollary~\ref{c:many-tokens}.

\item Additionally, if the message received by a node is associated with token $\tau$ which is larger than the minimum token id received in the previous executions of Smallest\_Token$(X)$,
(corresponding to previous rounds of the original \SaSmod\ protocol) then the node skips this message; if it is equal, the node continues the \SaSmod\ execution (for the same token id $\tau$); finally, if it is smaller (than all previously received token ids), the node abandons its current
execution of \SaSmod\ (associated with the previous, and bigger, token id) and joins a new one for the recently received smallest token id $\tau$ (i.e., the node assumes that this is the first time it receives a message associated with traversal of token $\tau$).
\end{itemize}
\item{Stage 3: Synchronization of termination time.} Once a root completes its execution of \SaSmod\ in Stage 2, i.e., when its token returns to it and there is no {\em unvisited} neighbour wrt the token issued by the root, the root initiates a simple Eulerian walk along the tree spanned by its token. During the walk, the token is propagated according to the computed variables $parent$
and lists $Child$ for the token. Additionally, a new variable counter is propagated with the token
to count the exact number $n$ of nodes. Note that the walk takes exactly $2n-2$ rounds.

  After completion, similar Eulerian walk is initiated, with only one change: instead of counting nodes, the already computed number $n$ is propagated together with the counter of rounds. This allows each node who receives this information to compute the exact time when this (second) Eulerian walk terminates. In this way, all nodes synchronize their termination round.
\end{description}

\begin{lemma}
\label{l:BTD-Traversals}
Algorithm BTD\_Traversals spans a \BTD\ tree rooted at some node that initially holds some rumors.
Additionally, all nodes terminate at the same round and the time complexity is $O(n\lg n)$.
\end{lemma}

\begin{proof}
  After Stage 1, which takes $O(\sum_i (2/3)^i n\lg n) \subseteq O(n\lg n)$ rounds, in the neighborhood of each survived node there is no other survived process. This is because the sequence of $(N,(2/3)^i n, (2/3)^i n/2)$-selectors, for $i=1,\ldots,\lg_{3/2} n$, 
guarantees that after the execution of the $i$-th selector there will be less than $(2/3)^i n$ active sources which
have not transmitted alone in the whole network (and thus were heard by all their neighbouring sources).
The proof is by straightforward inductive argument on $i$. Thus, after the execution of 
the last selector, there will be less than $n/n=1$ active sources which did not transmit alone.
Hence no two sources that survive by then, i.e., by the end of Stage 1, could be neighbors,
because in such case they they would have heard each other during the stage and could not have survived
by the end.
Hence, in the beginning of Stage 2 there is at most one token holder at each box of the pivotal grid
(otherwise they would be neighbors, which we just showed to be impossible).

During Stage 2, consider the execution of the smallest token. Observe that messages associated with this
token are always delivered, because by Corollary~\ref{c:many-tokens} 
there are never two tokens in the same grid box and the smaller always wins. Hence, by Lemma~\ref{l:SaSmod}
and the fact that each original round of procedure \SaSmod\ is emulated by $O(\lg n)$ rounds
of $(N,c)$-SSF,
Stage 2 finishes successfully in spanning a \BTD\ tree by the smallest token in $O(n\lg n)$ rounds.

Finally, the root of the spanned tree succeeds to synchronize all the nodes in two executions
of standard Euler walks along the tree, which takes $O(n)$ rounds.
\end{proof}

\paragraph{Algorithm BTD\_MB}
  Assume that all nodes are synchronised and they know the adjacent edges of the \BTD\ tree computed
by algorithm BTD\_Traversals and the exact number of nodes $n$. Every internal node initiates a stack to keep rumors to be forwarded to neighbouring \BTD\ nodes during the algorithm. The algorithm proceeds in two stages.

\begin{description}
\setlength{\leftskip}{-1cm}
\setlength{\itemindent}{-0.4cm}
\item{Stage 1: Sending rumors from leafs to internal nodes.}
  The root of the \BTD\ initiates an Eulerian walk along the tree by using token. Whenever the token visits a leaf who has some rumors, the leaf freezes the token and keeps transmitting
its rumors, one after another, addressed to its parent in the tree. The parent stores them upon receiving.

  After the walk terminates at the root, the root initialises another Eulerian walk, without freezing, carrying the number $n$ of nodes and the round counter, so that all nodes that receive this info know exactly when this walk will terminate. In this way, all nodes finish Stage 1 at the same time.

\item{Stage 2: Propagating rumors by internal nodes.} Every internal node, i.e., one who has at least one child in the tree, does the following. After receiving a new rumor, it puts it at the top of the stack. If only the stack is not empty, the node takes the top rumor and executes its transmission schedule defined by the $(N,c)$-SSF from Lemma~\ref{l:many-tokens}, sending the picked rumor whenever transmitting.
\end{description}

\begin{theorem}
Execution of algorithm BTD\_Traversals followed by BTD\_MB accomplishes multi-broadcast
task in $O((n+k)\lg n)$ rounds.
\end{theorem}

\begin{proof}
Observe that, by Lemma~\ref{l:BTD-Traversals}, after execution of algorithm BTD\_Traversals all the assumptions made in the beginning of the specification of algorithm BTD\_MB hold. The lemma also gives the time complexity $O(n\lg n)$ of algorithm BTD\_Traversals. In the remainder we focus on algorithm BTD\_MB.

  Stage 1 ``pulls'' all rumors from leaves to their parents (i.e., internal nodes) in $O(n+k)$ rounds. It also terminates synchronously, due to additional synchronising Euler walk.

  In the beginning of Stage 2, every rumor is stored in some internal node, and it is by definition
on the stack at that node. Since by Lemma~\ref{l:SaSmod-prop} every node has a constant number
of internal nodes in its neighborhood, the $(N,c)$-SSF from Lemma~\ref{l:many-tokens} assures
that each transmitted rumor (from the top of the stack) by a neighbour being an internal tree node
is successfully delivered during one run of the $(N,c)$-SSF. After such run, it is removed from the stack.

  Hence, a standard argument proves that after $O(n+k)$ runs of the $(N,c)$-SSF, every rumor is
delivered to every node that is a neighbour of some internal node (here $n$ is the upper bound
on the distance between the node holding the rumor in the beginning of Stage 2 and any other node,
while $k-1$ is the maximum number of other rumors that might be given priority before the 
considered rumor when stored in stacks). Combining this formula with the length of each execution of
the $(N,c)$-SSF, which is $O(\lg n)$, we obtain the final time complexity of Stage $2$.
\end{proof}

\bibliographystyle{abbrv}
\bibliography{references}
%
%

\end{document}